\documentclass[twocolumn]{article}

\usepackage[english]{babel}
\usepackage[utf8]{inputenc}
\usepackage[T1]{fontenc}

\usepackage[boldsans]{ccfonts}
\usepackage{eulervm}

\usepackage[%
bibencoding=utf8,
backend=biber,
bibstyle=ieee,
citestyle=numeric-comp,
sorting=nty,
natbib=true
]{biblatex}

\usepackage{amsmath}
\usepackage{amssymb}
\usepackage{amsthm}

\usepackage[%
 breaklinks,
 bookmarks,
 colorlinks=false,
 pdfborder={0 0 0},
 pdflang={en},
 pdfpagelabels=true,
 plainpages=false,
 unicode=true
] {hyperref}
\usepackage{hyperxmp}
\usepackage{ellipsis}

\usepackage[%
 capitalise,
 nameinlink,
 noabbrev,
 ] {cleveref}

\usepackage{csquotes}
\usepackage{algorithm2e}

\usepackage{tikz}
\usetikzlibrary{plotmarks, positioning}
\usetikzlibrary{calc,shapes,snakes,arrows.meta,chains}
\usepackage{pgfplots}
\defcitealias{hoefler-corrected-gossip}{Corrected Gossip}

\newcommand\reduce{{{reduce}}}
\newcommand\Reduce{{{Reduce}}}
\newcommand\allreduce{{{allreduce}}}
\newcommand\Allreduce{{{Allreduce}}}

\newcommand\inoperationfailures{{in-operation}{}}
\newcommand\Inoperationfailures{{in-operation}{}}
\newcommand\inoperational{{in-operational}}
\newcommand\inoperationally{{in-operationally}}

\newcommand\preoperationfailures{{pre-operation}{}}
\newcommand\preoperational{{pre-operational}}
\newcommand\preoperationally{{pre-operationally}}
\newcommand*{\defeq}{\mathrel{\vcenter{\baselineskip0.5ex \lineskiplimit0pt
                     \hbox{\scriptsize.}\hbox{\scriptsize.}}}%
                 =}
\newcommand*{\N}{\mathbb{N}}
\newcommand*\comma{,}
\DeclareMathOperator{\xmod}{mod}

\newtheorem{theorem}{Theorem}

\newtheorem*{definition}{Definition}

\begin{document}

\title{Fault-tolerant Reduce and Allreduce operations based on correction}

\author{Martin K\"uttler \\ martin.kuettler@tu-dresden.de \\ TU Dresden \and Hermann H\"artig \\ hermann.haertig@tu-dresden.de \\ TU Dresden}

\maketitle

\section{Introduction}

When multiple {processes}, i.e. independently executing entities,
potentially executing on distinct hardware, collaborate to achieve a
common goal they require communication. One prominent example of this
scenario is \emph{{HPC}}.

Beyond one to one communication (also
called \emph{Point-to-Point-communication}), that only involves two
communication partners, important operations to facilitate this
communication are so called {\emph{collective
(communication)}}\emph{ operations}, where a potentially bigger set of
processes exchanges data following some communication
pattern~\cite{chunduri18characterzation_of_mpi_usage}.

Commonly used examples of such communication patterns are broadcast,
where one process sends the same data to all other
participants, reduce, where all processes contribute some data and
one process receives the combined result, and allreduce, which works
like reduce but all processes receive the result. These operations
are examples for the more general communication patterns one to many,
many to one, and many to many.

It is common to replace `many' with `all.'  Usually a different
technique is used for restricting the set of processes participating,
e.g., {MPI} groups\footnote{A MPI communicator, that is directly used to
achieve this, uses a MPI group to store membership.}~\cite{mpi}. This
point of view is adopted in this document: It is always assumed that
all processes take part in the collective communication operations.

Point-to-point messages are a primitive operation of computer
networks. {Collective} operations can be implemented using
that primitive, i.e., by sending multiple one to one messages between
selected processes. The goal of this document are algorithms that
implement a {collective} operation using this primitive.

For example, broadcast can be implemented by sending messages along
the edges of a tree that is rooted at the original sender of
the broadcast (the \emph{root} process). Processes correspond to
vertices in the tree, and each process sends the data to all its
children, after receiving the data (if the vertex has a parent in the
tree).

This example shows how process {failures}, where a process fails
to send some or all messages pertaining to the {collective}
algorithm, impede the {collective} operation. If in the tree
one {process} does not send messages to the {processes}
corresponding to its child vertices, all subtrees rooted at its
children do not receive any data.

The focus of this work are fault tolerant algorithms
for reduce and allreduce operations with small messages. Failures of
communicating processes are considered exclusively, in contrast to
missing/delayed/otherwise impeded messages.

Algorithms for reduce, and allreduce are presented in this paper. An
algorithm for broadcast was previously
described~\cite{kuttler-corrected-trees}.

The semantics of fault tolerance aimed for in this paper is: If not
more than a number $f$ of {processes} fail preoperationally
(i.e. before the communication operation), the {collective}
operation provides the same result as if the failed processes were
excluded by all participating processes in advance. Processes
failing inoperationally (during the operation) can appear either
alive or dead with respect to the operation, but no mixed state is
allowed.

The algorithms target small messages. For big messages, they still
behave correctly, but other implementations are more efficient. Which
sizes constitute a small or big message cannot be stated in general,
and is not discussed here. A rough distinction can be made based on
whether the operation is latency critical (the focus of this work), or
bandwidth critical.

\section{Related work}

Correction was first used by
\citeauthor{hoefler07practically}~\cite{hoefler07practically} in
combination with an unreliable hardware multicast. It was intended to
be executed after a hardware multicast, to correct for potential
omissions, in order to turn the operation into a broadcast. However,
only probabilistic guarantees are given.

The basic purpose of up-correction is the same: to correct omissions
made by the fault-agnostic algorithm of reduce. It the context of
this work, omissions are not random, but deterministically caused by
process failures.

The idea of correction was significantly expanded upon by
\citeauthor{hoefler-corrected-gossip}~\cite{hoefler-corrected-gossip}
in the work on \citetalias{hoefler-corrected-gossip}. Three different
correction algorithms with different semantics are proposed. They
are all intended to be combined with a gossip phase that disseminates
a broadcast value probabilisticly. Following this step,
correction is used to improve the probabilities provided by gossip,
or for some correction algorithms even give some guarantees, under
assumptions on the number and timing of process failures.

The up-correction algorithm described in \cref{upcorrection} are based
the same ideas as the correction algorithms used
in \citetalias{hoefler-corrected-gossip}. They serve a different
purpose, however. In \citetalias{hoefler-corrected-gossip}
the correction algorithms are used to help against the inherent
shortcoming of gossip, i.e., that messages are sent randomly and thus
some processes might never receive a message. Their use has no
relation to fault tolerance there. Instead, disseminating information
by sending messages to random participants is inherently tolerant to
crash {process} {failures}. In this document, the
up-correction algorithm are used to correct omissions induced by
process failures.

In addition to the theoretical difference of the purpose of
correction, there is a practical difference
between \citetalias{hoefler-corrected-gossip}'s correction and
up-correction as used in this
work. In \citetalias{hoefler-corrected-gossip}, the gossip phase, and
the following correction phase are assumed to be two globally separate
phases. The published algorithm is only simulated, not practically
implemented. Conversely, in this work the active phase is a local, not
a global, property: {processes} execute the phases in succession,
but independently of the progress of other {processes}.

\section{Failures}

In this document failures of processes are assumed to follow
a \emph{{fail-stop}}-model, i.e., processes that fail stop sending
any messages. Sending to a failed process does not provide any
indication of the failure. Instead, the send operation completes like
a send operation to a live process, and the failed process will not
provide any reaction to the message. This is also sometimes called
a \emph{crash} failure model. Distinctions between the two terms are
made inconsistently, and none are assumed here.

The network is assumed to be reliable, i.e., messages are not lost,
reordered, or modified. In other words, no communication failures are
possible.

\section{Reduce}

Without loss of generality it is assumed that the recepient of the
reduce (i.e., the root) is process $0$. If this is not the case, its
number can be swaped with that of process $0$ to restore this
property.

The basic reduction function is assumed to be assocative (as, e.g.,
mandated by MPI~\cite{mpi}) and commutative for the description
below.

The following terminology is used. A reduce operation is identified by
a reduce message $m$. The contents of this message are defined
below. Processes do not need to receive a message before sending
one. Depending on the implementation, some {processes} need to receive
a message first, but logically most processes just contribute a value,
and do not receive any information. Thus, reduce messages for the same
operation are created by multiple {processes}.

Before sending any network messages related to $m$, a process
indicates its intention by calling \emph{init\_reduce($m$)}. After the
reduce-operation is completed locally, i.e., for a non-root process after
sending all information to the parent process, and for the root after
the final result is available the process reports to the caller by
calling \emph{deliver\_reduce($m$)}.

A reduce message has the following data members:

\begin{itemize}

\item (A descriptor of) the set of participating processes.

\item A unique id.

\end{itemize}

\subsection{Semantics}
\label{reduce-semantics}

Let $n$ {processes} execute reduce. Up to $f$ arbitrary processes
may fail either \preoperationally\ or \inoperationally.

\begin{enumerate}

\item If the root calls {\emph{deliver\_reduce($m$)}}, all
  non-failed processes have called {\emph{init\_reduce($m$)}}.

\item Each process calls {\emph{deliver\_reduce($m$)}} at most once
  for a given $m$.

\item The value returned by reduce at a non-failed root includes the
  input values of all non-failed processes.

\item Values of {failed} {processes} are either
  included normally, or disregarded. No intermediate stages are
  possible.

\item Each non-failed process executes \emph{deliver\_reduce($m$)} eventually,
  if all non-failed processes execute \emph{init\_reduce($m$)}.

\end{enumerate}

Reduce sent to a failed process simply becomes a no-op.

\subsection{Up-Correction}
\label{upcorrection}

The underlying idea of the algorithm for reduce presented in this
document is to combine a tree phase with a preceeding up-correction
phase. The tree phase implements the reduce operation completely in
the failure-free case.

In case a {process} fails, the dissimination of data is inhibited. The
tree phase consists of every process receiving messages all children,
followed by every process except the root sending a single messages to
its parent in the tree. This dissimination of data is inhibited in
case there is a faulty procss that does not send a message to its
parent.

\begin{algorithm}
  \DontPrintSemicolon
  \KwIn{The data contributed by this process}
  \KwOut{The data used in the tree phase by all processes in this
    {up-correction group}}
  \SetKwData{Data}{data}
  \SetKwData{ProcessId}{process\_id}
  \SetKwData{Group}{group}
  \SetKwFunction{ReplicateFct}{up\_correction}
  \SetKwProg{Fn}{function}{ begin}{end}
  \Fn{\ReplicateFct{\Data}}{
    \Group $\leftarrow$ compute\_up\_correction\_group(f, \ProcessId)\;
    senddata $\leftarrow$ \Data\;
    \For{$p \in $ \Group $\setminus\{\ProcessId\}$}{
      \emph{Note: no failure information is sent here}\;
      send (senddata) to $p$\;
      receive (data) from $p$\;
      \Data $\leftarrow$ reduce\_function(\Data, received data)\;
    }
    \KwRet \Data\;
  }
  \label{alg:upcorrection}
  \caption{Up-correction}
\end{algorithm}

To compensate for this problem, {processes} exchange their data in
\emph{{up-correction groups}} in the up-correction phase prior to
the tree phase.

To tolerate up to $f$ failures, all processes $p$ that share the
\emph{group number} $\left\lfloor \frac{p - 1}{f + 1}\right\rfloor$
form one \emph{{up-correction group}} and exchange messages. In
addition, if the last group (the one with the highest number) has less
than $f + 1$ members, the root is also part of it. Otherwise, the root
does not belong to any group. Values are exchanged with (sent to, and
received from) each other member in the same group.

Values are exchanged in the group, and reduced, i.e., combined using
the reduction function, locally. The result is that all members of an
{replication group} have the same value after their
up-correction phase (\Inoperationfailures\ make this situation more
complicated. If members of a group experience \inoperationfailures,
the final value of other members of the group might or might not
include the value of the failed process. Both results are considered
correct in the semantics of reduce. This value is used in the
tree phase. Details are in the description of reduce, see
\cref{reduce-description}.

The motivation of the design of the groups is that {processes} with the
same number in the subtrees of the root exchange
values. This limits the number of messages that are sent and received
by each process, and enables the root to efficiently reason about
which value is included in which subtree.

Note that in an {up-correction group} that includes at least one failed
{process}, all live processes will time out on the respective receive
operations and will confirm the sender to have failed with the
respective failure monitor. The resulting delay is unfortunate, but
not avoidable: In reduce, where every process contributes a value,
each process that fails to send a value must be confirmed to have
failed. How this is done is independent of the communication
algorithm. Timeouts are used here.

\subsection{Algorithm}
\label{reduce-description}

\begin{algorithm}
  \DontPrintSemicolon
  \KwIn{The data contributed by this process}
  \KwOut{The resulting data}
  \SetKwData{Data}{data}
  \SetKwData{Root}{root}
  \SetKwFunction{ReduceFct}{reduce\_root}
  \SetKwProg{Fn}{function}{begin}{end}
  \Fn{\ReduceFct{\Data}}{
    \If{root is in a up-correction\ group}{
      \Data $\leftarrow$ \ReplicateFct(\Data)\;
    }
    \BlankLine
     msg, $c$ $\leftarrow$ receive from any child $c$ (data, failure information)\;
     \If{failure information indicates failure in subtree}{
       continue\;
     }
     \eIf{root data not included in data from sender (determined from sender id)}{
       \KwRet basic\_reduce\_function(\Data, received data)\;
     }{
       \KwRet received data\;
     }
     raise Error(``No failure-free subtree'')\;
  }
  \label{alg:reduce-root}
  \caption{Fault tolerant \reduce, algorithm executed at the root}
\end{algorithm}

\begin{algorithm}
  \KwIn{The data contributed by this process}
  \DontPrintSemicolon
  \SetKwData{Data}{data}
  \SetKwData{Root}{root}
  \SetKwFunction{ReduceFct}{reduce\_non\_root}
  \SetKwProg{Fn}{function}{begin}{end}
  \Fn{\ReduceFct{\Data}}{
    \Data $\leftarrow$ \ReplicateFct(\Data)\;
    \BlankLine
      \For{$c \in$ children}{
        receive (data, failure information) from $c$\;
        \Data $\leftarrow$ basic\_reduce\_function(\Data, received data)\;
        update failure information\;
      }
      send (\Data, failure information) to parent\;
  }
  \label{alg:reduce-non-root}
  \caption{Fault tolerant \reduce, algorithm executed at non-roots}
\end{algorithm}

\begin{algorithm}
  \KwIn{The data contributed by this process, and the number of the root
    process}
  \KwOut{At the root: The resulting data}
  \DontPrintSemicolon
  \SetKwData{Data}{data}
  \SetKwData{Root}{root}
  \SetKwData{ProcessId}{process\_id}
  \SetKwFunction{ReduceFct}{reduce}
  \SetKwFunction{ReduceRootFct}{reduce\_root}
  \SetKwFunction{ReduceNonRootFct}{reduce\_non\_root}
  \SetKwProg{Fn}{function}{begin}{end}
  \Fn{\ReduceFct{\Data, \Root}}{
    \eIf{\ProcessId == \Root} {
      \Return{\ReduceRootFct{\Data}}\;
    }{
      \ReduceNonRootFct{\Data}\;
    }
  }
  \caption{Fault tolerant \reduce, that calls either version above}
  \label{alg:reduce}\label{alg_reduce}
\end{algorithm}

{Processes} that {fail} before calling a \reduce\ operation
do not contribute their input value. This can not be avoided. The
purpose of this fault tolerant operation is to make sure that
these {processes} do not hinder other processes from contributing
their input value.

For the sake of the description, it is assumed that the
root process does not fail. No workaround for
the case of a failed root is necessary. If the root fails
(\preoperational\ or \inoperational), this operation becomes a no-op,
adhering to the semantics.

{Processes} first execute the up-correction\ algorithm, and then
execute a tree\ phase, in which a normal \reduce\ up a tree is
performed.

The up-correction\ phase is described in \cref{upcorrection}. The
resulting reduced value is then used in the
tree\ phase to be delivered to the parent {process}.

In the failure-free case, the root {process} gets values from all its
children that in total (including the value of the root itself)
include each value $f+1$ times. By selecting one value it received,
the root gets the correct value. Which value the root should
select, based on the information on failures, is explained
in \cref{reduce-failure-information}. There are multiple ways for
failure information to be propagated.

\begin{figure}
    \begin{center}
\begin{tikzpicture}[
every node/.style = {draw,circle,minimum size=6mm,inner sep=0mm},
        level 1/.style={sibling distance=49mm},
level 2/.style={sibling distance=23mm},
level 3/.style={sibling distance=9mm},
level/.style={level distance=10mm}]
\node (n0) {0}
  child { node (n1) {1}
    child { node (n3) {2} }
    child { node (n5) {3} }
  }
  child { node (n2) {4}
    child { node (n4) {5} }
    child { node (n6) {6} }
  };

  \node[cross out, draw=red] at (n1.center) () {1};
  \draw (n2) edge [bend right=30,->]
             node [draw=none, midway, label=right:4--6] {}
        (n0) ;

  \path (n3) edge [bend left=30,->]
             node [draw=none, midway, label=left:2] {}
        (n1);
  \path (n5) edge [bend right=30,->]
             node [draw=none, midway, label=right:3] {}
        (n1);

  \path (n4) edge [bend left=30,->]
             node [draw=none, midway, label=left:5] {}
        (n2);
  \path (n6) edge [bend right=30,->]
             node [draw=none, midway, label=right:6] {}
        (n2);
\end{tikzpicture}
    \end{center}
\caption{The failed process 1 impedes the propagation of data in the
tree phase. Arrow labels show the values of which processes
are included in the respective message along the path.}
\label{fig:failure_reduce_tree}
\end{figure}
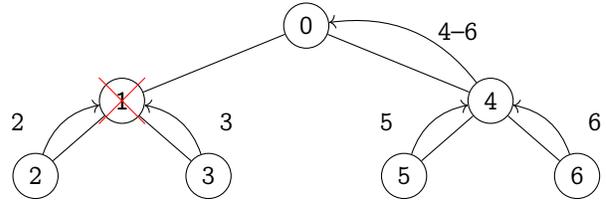

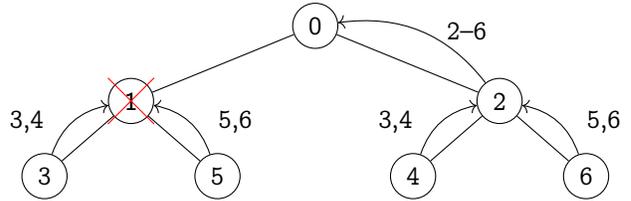
\begin{figure}
  \begin{center}
\begin{tikzpicture}[
every node/.style = {draw,circle,minimum size=6mm,inner sep=0mm},
        level 1/.style={sibling distance=49mm},
level 2/.style={sibling distance=23mm},
level 3/.style={sibling distance=9mm},
level/.style={level distance=10mm}]
\node (n0) {0}
  child { node (n1) {1}
    child { node (n3) {3} }
    child { node (n5) {5} }
  }
  child { node (n2) {2}
    child { node (n4) {4} }
    child { node (n6) {6} }
  };

  \node[cross out, draw=red] at (n1.center) () {1};
  \draw (n2) edge [bend right=30,->]
             node [draw=none, midway, label=right:2--6] {}
        (n0) ;

  \path (n3) edge [bend left=30,->]
             node [draw=none, midway, label=left:3\comma4] {}
        (n1);
  \path (n5) edge [bend right=30,->]
             node [draw=none, midway, label=right:5\comma6] {}
        (n1);

  \path (n4) edge [bend left=30,->]
             node [draw=none, midway, label=left:3\comma4] {}
        (n2);
  \path (n6) edge [bend right=30,->]
             node [draw=none, midway, label=right:5\comma6] {}
        (n2);

  \useasboundingbox (n3.west) -- (n0.north) -- (n6.east) -- (n6.south);

\end{tikzpicture}
  \end{center}
\caption{Up-correction phase and subsequent tree phase with the
  same failed processes as in \cref{fig:failure_reduce_tree}.}
\label{fig:failure_reduce_tree_up-correction}
\end{figure}

As an example of \reduce\ with a {failed} {processes}, consider seven
processes that want to compute the sum of their process numbers, or
rank numbers in MPI terms. The process number is a unique identifier
in the range $0$ to $6$. Process $1$ is assumed to have failed. The
goal is thus to compute the sum of the remaining process numbers
$0+2+3+4+5+6=20$. Process $0$ is the root of the
\reduce\ operation. \Cref{fig:failure_reduce_tree,fig:failure_reduce_tree_up-correction}
depict the communication sent in this scenario in a ``common'' tree
implementation, and in the tree\ phase of the algorithm presented
here. The algorithm constitutes an up-correction\ phase and a
tree\ phase. The labels of the arrows in
\cref{fig:failure_reduce_tree,fig:failure_reduce_tree_up-correction}
depict the process ids whose values are included in the (partial) sum
sent in the respective message (this information is not sent, and is
only included in the picture for illustrative purpose). Since each
process sends its process id as a value in this example, the sum of
the numbers listed is the value being sent in each message.

In the tree algorithm, which corresponds to the tree\ phase, the
processes at the leaf positions in the tree start by sending their
process number to their parents. I.e., processes $2$, $3$, $5$, and $6$
send messages. Next, process $4$ receives the messages sent to it,
combines the two corresponding values with its own value, and obtains
the resulting partial result $5 + 6 + 4 = 15$, which it sends to its
parent. Unfortunately, process $1$ does not act likewise, because it
has failed. Finally, process $0$ receives the partial result $15$, and
adds it to its local value of $0$. It does not receive any
contribution from the left subtree.

Things go differently when an up-correction\ phase is first
executed. This is depicted
in \cref{fig:failure_reduce_tree_up-correction}. Note that the numbering
in the tree is different now. While it was arbitrarily chosen to be
depth-first in \cref{fig:failure_reduce_tree}, the numbering is now
matching the numbering scheme for \reduce, as described
in \cref{upcorrection}. Process $0$, the only one with a special role,
did not change its position in the tree.

In this example there are $f = 1$ failed processes, and in
the up-correction\ algorithm processes exchange messages in groups of
$f + 1 = 2$. Processes $3$ and $4$ send messages to each other, as do
processes $5$ and $6$. Processes $1$ and $2$ would too, but since
process $1$ has failed, it does not send a message. Process $0$ does
not send any messages in this case, because $n-1=6$ is divisible by
$f+1=2$, and thus process $0$ is not a member of any {up-correction
group}. Upon receiving a message, each process updates its local value
with the one sent to it. I.e., processes $3$ and $4$ hold the value $3
+ 4 = 7$ afterwards. Processes $5$ and $6$ store $5 + 6 =
11$. Processes $2$ and $0$ do not receive any message, and retain
their original values.

Next, the tree\ phase works like the tree algorithm described
above, but all processes send the value obtained in the up-correction\
phase instead of their original value. Initially, processes $3$ and
$4$ send $7$, and processes $5$ and $6$ send $11$. Next, process $2$
receives the two messages sent to it, and computes the sum of their
values and its own local value $7 + 11 + 2 = 20$. This result is sent
to the root. Process $1$ does not send any messages. Process $0$
receives one message with the value $20 = 2 + 3 + 4 + 5 + 6$, and adds
its local value $0$ to it, to get the final result, $20$. Like before,
process $1$ did not get to contribute any value, because it failed
before the operation. Unlike before, the children of process $1$ did
get to contribute their values. Since process $2$ can include the
information that there where no failures in its subtree, process $0$
knows that the value it received is complete.

Generally, the up-correction\ phase (described in \cref{upcorrection}),
yields a single value, $v$, in each {process}. The root is an
exception if it is not part of the last group. If it is not, its value
$v$ is its input value to \reduce. For all other processes, $v$ is the
result of reducing their input value with the input values of all
other processes in the same up-correction\ group. The resulting value
$v$ is used in the tree\ phase. Each process, except for the
root, waits for messages from all their children, reduces their value
with the local $v$, and then sends the result to its parent. This is
shown in pseudocode in algorithm \ref{alg_reduce}.

If only the tree\ phase were executed, the root would get the
results from combining the results from all subtrees of its child
processes, if no failures occur. If up-correction\ is executed before,
the root gets from each child process either the complete result, or
the complete result without the value of the root, or the complete
result without the values of the last group. The values of the last
group are included if the last group has a member in the respective
subtree. If root is a member of this group, its value is naturally
included. In any case the root {process} has the information
necessary to complete the result it gets from its children. Thus, the
first answer that includes an indication that no failure happened in
the respective subtree (more on that
in \cref{reduce-failure-information}) suffices for the root to compute
the final result, and return from the call to \reduce.

\subsection{Failure information}
\label{reduce-failure-information}

The root receives multiple
results. In the {failure}-free case, a correct result is sent by all
children of the root, and the root can select any one. If {processes}
fail, there can be incorrect results, and the root needs some
information based on which it can select a correct value. Note that,
corresponding to the semantics in \cref{reduce-semantics}, different
correct results are possible. If processes fail \inoperational, no condition
is made on whether they get to include their value. It is possible
that both results -- one including the value from an eventually
failing process, and one without it -- are sent to the root along
different paths. This can, e.g., happen if a process completes the
up-correction phase, but does not execute the subsequent
tree phase.

If only \preoperationfailures\ occur, there is only one correct
result.

To enable the root to select a valid result, a failure description is
accumulated in each subtree, and sent along with the reduction
value. There are multiple options for this information.

The failure description that provides the most information about
failures to the root, at the cost of the highest potential message
size, is sending a list of known failed processes. Note
that \inoperational\ failures can lead to a process appearing non-failed
(contributing its value, and not being included in the list of failed
processes), when in fact it fails shortly after. This can not be
avoided, and it will potentially be detected in the next communication
operation spanning this process (or never, if there is no such
operation).

In this scheme, in the up-correction\ phase and in the tree
phase, each {process} appends the ids of all processes it can not
receive a value from to a list of failed processes. This list is
appended to the message that is sent to the parent in the tree\
phase and the parent adds the lists of its children to its own. In that
way, the root receives a complete list of failed processes from each
child that is known to have the complete reduction information, i.e.,
that does not contain a failed process in the subtree spanned by it.

One potential use of the list of failed processes is to make that
information available to all processes, to exclude failed processes in
future operations. This is not described here further.

A simplification of this scheme, that requires less data to be sent,
at the expense of less information being available, is to only send
the size of this list. Since the list is only appended to, without any
regards to its content (note that items in lists that are concatenated
always come from disjoint sets), the size can easily be tracked. To
enable the root to choose which subtree's value to select, each
process needs to send along a bit to indicate whether a process failed
in this subtree. This bit is set when a child is found to have
failed in the tree.

A third -- even simpler -- scheme is to only send a single failed
bit. This bit is set in the tree phase, if a process does not
receive a value from one of its children. It is not modified in the
up-correction phase. The bit is equal to the 'local' bit in the
second scheme.

That way, the root {process} only knows whether some {failure}
happened in a subtree.

\subsection{Properties}

\label{reduce-proofs}

The goals of this section are to proof the semantics described
in \cref{reduce-semantics}, and to detail the number of messages that
are required for \reduce.

\begin{definition}
  Let $T$ be a tree, $f\in\N_0$. Let $r$ be the root of $T$. Then $T$ is an
  \emph{{I(f)-tree}}, if

  \begin{enumerate}

  \item $r$ has $f+1$ children $c_0,\ldots,c_f$. The subtrees spanned
    by these vertices are called $T_0,\ldots,T_f$.

  \item For any two subtrees of the root $T_i, T_j$, $i,j\in \N_0$,
    $i,j \le f$, the sizes of $T_i$ and $T_j$ differ by at most one.

  \end{enumerate}

\end{definition}

\begin{theorem}
  \label{thm:reduce-replication} Let there be no more then
  $f$ {processes} that experience a {failure}, \inoperational\
  or \preoperational. Let a {I(f)-tree} be used,
  with {up-correction groups} of size $f+1$.

  After up-correction, all values of non-failed processes, except for the
  values of processes grouped with root, are included exactly once in
  the final value available in each subtree of the root that does not
  indicate a failure after the tree phase.
\end{theorem}
\begin{proof}
  First it is shown that the value from an arbitrary process is
  included at least once in each subtree. This is shown by
  constructing a process that has the value. Afterwards it is argued
  that there can not be multiple processes in one subtree that hold
  the value.

  Let $\ell$ be a non-failed process, $k\in\N_0$, $k > 0$, $k \le f+1$. It
  needs to be shown that the data from process $\ell$ is available in
  the $k$-th subtree, which is the subtree spanned by process numbered
  $k$.

  If $\ell = 0$, it is the root, so it is considered ``grouped with
  the root,'' and there is nothing to show.

  Let $\ell$ not be grouped with the root, i.e.,
  $\left\lfloor\frac{\ell - 1}{f+1}\right\rfloor < \left\lfloor\frac{n
  - 1}{f+1}\right\rfloor$ and $\ell\ne0$. In up-correction $\ell$
  exchanges data with the process numbered
  $a \defeq \left\lfloor\frac{\ell - 1}{f+1}\right\rfloor (f+1) + k$
  (or $a$ is the number of $\ell$).

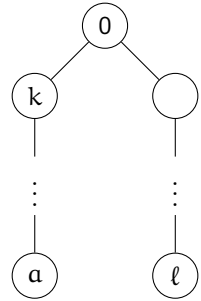
  \begin{minipage}{0.4\linewidth}
    \begin{tikzpicture}
      \node[draw,circle,minimum size=6mm,inner sep=0mm] (root) {$0$};
      \node[draw,circle,below left=5mm and 5mm of root,minimum size=6mm,inner sep=0mm] (nodek) {$k$};
      \node[draw,circle,below right=5mm and 5mm of root,minimum size=6mm,inner sep=0mm] (node_parent) {$ $};
      \draw (root) -- (nodek);
      \draw (root) -- (node_parent);
      \node[below=5mm of nodek] (treek) {$\vdots$};
      \node[below=5mm of node_parent] (treeell) {$\vdots$};
      \node[below=5mm of treek,draw,circle,minimum size=6mm,inner sep=0mm] (nodea) {$a$};
      \node[below=5mm of treeell,draw,circle,minimum size=6mm,inner sep=0mm] (nodeell) {$\ell$};
      \draw (nodek) -- (treek);
      \draw (treek) -- (nodea);
      \draw (node_parent) -- (treeell);
      \draw (treeell) -- (nodeell);
    \end{tikzpicture}
\end{minipage}%
\begin{minipage}{0.57\linewidth}
  Then $a < n$, and the process
  numbered $a$ is in the $k$-th subtree, because $(a - 1) \xmod (f+1)
  = k - 1$. If the process numbered $a$ does not fail at least long
  enough to send a message to its parent in the tree phase, its
  value is included in this message. Barring any failures along the
  way, the value is going to be included in the final result of the
  subtree. Thus the proof is complete in this case.  \end{minipage}

  If the process numbered $a$, or any of its parents, experience a
  failure, and do not get to send a message in the tree phase,
  their respective parent detects the failure and propagates that
  information upward. Accordingly, the root of the subtree (i.e., the
  process numbered $k$) will not report no failures in its subtree.

  Note that additional {failures} can stop the propagation of
  this failure information. These additional failures will be
  detected, however, and will lead to process number $k$ not claiming
  ``no failures.''

  This shows that there is a process in each given subtree that holds
  the value from process $\ell$. It remains to be shown that there is
  at most one.

  To see that at most one {process} in a different subtree
  receives the value from a given process, note that Communication for
  any process except for the root only occurs with at most one process
  from each other subtree. That holds because each process belongs to
  at most one {up-correction group}, and all communication except for
  the messages between the root of the tree and its children in
  the tree phase is within one subtree.
\end{proof}

\begin{theorem}
  \label{thm:reduce} Let there be no more then $f$ processes that
  experience a {failure}, \inoperational\
  or \preoperational. Let a {I(f)-tree} be used,
  with {up-correction groups} of size $f+1$.

  In the tree phase, all children of the root will either yield
  a correct value, or indicate a {failure} in their subtree, or
  have failed themselves.
\end{theorem}
\begin{proof}
  In the tree phase, each non-failed {process} collects the messages
  from all its children and reduces the value received before sending
  a message itself. Thus, by \cref{thm:reduce-replication}, subtrees
  without failures propagate all data. The value of (all) processes
  grouped with the root is included iff one such process is in the
  subtree in question.

  Now assume a {failure} does occur in the subtree in
  question. For ease of wording, assume first that the root of the
  subtree, i.e., the child of the root, has not failed.  A failure in
  a subtree is guaranteed to lead to a failure being reported in the
  failure information sent together with the value by the root of the
  subtree, by the following logic. In the tree phase any
  process {failure} will eventually be detected by a non-failed parent,
  and be propagated by further non-failed parents. A failed parent will not
  propagate the failure information, but the next non-failed {process} will
  detect a {failure}.

  In this way, if the root of a subtree has not failed, it will report
  correctly on the presence of failures in its subtree.

  In case the respective child of the root failed, the theorem is
  trivially satisfied. The root will detect the failure, and classify
  the corresponding subtree as containing a failure.
\end{proof}

When the root receives values from its children, it must select one of
them, and potentially combine it with its own value. By
\cref{thm:reduce}, it can select the value from any subtree
without a failure. It knows which these are because all schemes
for failure information propagation described before contain that
information.

\begin{theorem}
  \label{thm:reduce-number-failures} Let there be no more then
  $f$ {processes} that experience a {failure}, \inoperational\
  or \preoperational. Let a {I(f)-tree} be used,
  with {up-correction groups} of size $f+1$.

  After \reduce\ either the root will have
  failed, or it will know the correct result, as described
  in \cref{reduce-semantics}.
\end{theorem}
\begin{proof}
  If no more than $f$ processes fail, at least one of the $f+1$
  children of root is guaranteed to have no failure in the subtree
  spanned by it.

  If the root has failed, the theorem is satisfied. Otherwise,
  by \cref{thm:reduce}, one child of the root will yield the correct
  value. The root might have to combine its value with its own value
  $v$, depending on which child reported it. The root can detect a
  valid child to choose the value from because it does not report a
  failure in its subtree.
\end{proof}

\begin{theorem}
  \label{thm:reduce-semantics} Let there be no more then
  $f$ {processes} that experience a {failure}, \inoperational\
  or \preoperational. Let a {I(f)-tree} be used,
  with {up-correction groups} of size $f+1$.

  \Reduce, as described
  in \cref{reduce-description}, satisfies the semantics
  from \cref{reduce-semantics}:
  \begin{enumerate}

  \item If the root calls {\emph{deliver\_reduce($m$)}}, all
        non-failed processes have called {\emph{init\_reduce($m$)}}.

  \item Each process calls {\emph{deliver\_reduce($m$)}} at most once
        for a given $m$.

  \item The value returned by \reduce\ at a non-failed root includes the
        input values of all non-failed processes.

  \item Values of {failed} {processes} are either
        included normally, or disregarded. No intermediate stages are
        possible.

  \item Each non-failed process executes \emph{deliver\_reduce($m$)} eventually,
        if all non-failed processes execute \emph{init\_reduce($m$)}.

  \end{enumerate}
\end{theorem}
\begin{proof}
  \begin{enumerate}

  \item This property follows from property 3. Note
    that {failed} {process} never calls
    {\emph{deliver\_reduce}}.

  \item \Reduce\ is delivered once after the tree phase. It can
    not be skipped. Thus multiple receives are not possible.

    A leaf process, that does not receive in the tree phase,
    could theoretically execute it multiple times. It does not do so
    without external trigger, however, and a new call to \reduce\
    leads to a new unique id in the reduce message.

  \item This follows from \cref{thm:reduce-number-failures}.

  \item Let $p$ be a {failed} {process}. If $p$
    failed \preoperationally, it does not send any messages in the
    operation. Accordingly, the value of $p$ is not known to any non-failed
    processes, and is not included.

    Assume $p$ fails \inoperationally. First assume that $p$ is not
    grouped with the root.  Let $S$ be the subtree spanned by a child
    of the root that the root eventually selects the value from
    (cfg. \cref{thm:reduce}). Then the result depends on the state of
    $p$ at the time of the sending operation to $S$. The relevant
    sending operations are either the sending to the member of
    the up-correction group in $S$, if $p$ is not in $S$, or the
    message sent by $p$ to its parent, if $p$ is in $S$.

    If $p$ does not fail before sending the respective message, its
    value is propagated normally, and is included in the result, like
    that of a non-failed process. If $p$ fails before sending that
    message, it is perceived like if it
    had {failed} \preoperationally. No value from $p$
    is included in this case.

    If $p$ fails \inoperationally\ and is grouped with the root (i.e., in the
    same up-correction group as the root), the reasoning is a little
    different. If $p$ is the root, its value is eventually
    disregarded, as nobody is going to receive the final message. If
    $p$ is not the root, either its message sent to the root in the up-correction\ phase plays the role
    of the message sent to a member of $S$, if $S$ has no common
    member with the up-correction\ group of $p$, or, if $S$ contains a
    member of the up-correction\ group of $p$, called $v$, the message
    from $p$ to $v$ plays that role.

  \item All that needs to be seen is that for some process $p$ the
    work between successive calls to \emph{reduce} and \emph{deliver}
    is bounded.

    First, in the up-correction\ phase, $p$ sends to and receives from
    all other {processes} in its {up-correction group}. These
    are up to $f$ processes. The receiving is retried if the failure
    monitor of the prospected sender confirms that it has not
    failed. Given that the {process} executes \emph{reduce}, it
    will eventually either fail, in which case $p$ stops waiting for
    the message, or send a message, in which case $p$ will receive.

    The same holds for the subsequent tree\ phase. All
    non-failed children will eventually send, and $p$ will stop
    waiting for receives. After that, unless $p$ is the root, it will
    send a single message to its parent.

  \end{enumerate}
\end{proof}

Note that in the presence of \inoperationfailures, item 4
of \cref{thm:reduce-semantics} allows for multiple different correct
results. The root {process} gets to select the final result.

\begin{theorem}
  \label{thm:reduce-number-messages}
  Let \reduce\ be executed without any {failures}. Then the number of
  messages sent between {computation processes} is the following.

  \begin{itemize}

    \item In the up-correction\ phase
      $f(f+1)\left\lfloor\frac{n-1}{f+1}\right\rfloor + a(a-1)$ with $a
      = ((n-1) \xmod (f+1)) + 1$ are being sent.

    \item In the tree\ phase $n-1$ messages are being sent.

  \end{itemize}

  This does not include any messages that might be necessary for the
  detection of {failures}. These can even be required when no failures
  occur.

  When processes fail, less messages are being sent.
\end{theorem}
\begin{proof}

  First consider the failure-free case.

  \begin{itemize}

  \item There are
    $\left\lfloor \frac{n-1}{f+1} \right\rfloor$ {up-correction groups}
    with $f+1$ {processes} each. In each such group each of the $f+1$
    processes sends $f$ messages. That accounts for the first term.

    If $a > 1$, it is the size of the last up-correction\
    group. $a(a-1)$ is the number of messages sent in this group,
    because each of the processes sends $a-1$ messages. If $a=1$,
    there is no additional {up-correction group}, but $a(a-1)=0$.

  \item In the tree\ phase each process except for the root
    sends one message to its parent in the tree.

  \end{itemize}

  If there are {failures}, the failing processes send less
  messages. Perhaps none, if the process failed \preoperational. No
  other {process} sends more messages based on failures.

\end{proof}

\section{Allreduce}
\label{allreduce}

Without failures, \allreduce\ works like \reduce, but the result is
provided at all processes. There is no root process in \allreduce.

\Allreduce\ is one of the most common operations in HPC
programs~\cite{grbovic2007performance}, and is frequently
studied~\cite{kolmakov2020generalization,%
patarasuk2007bandwidthefficient,zhae2014kylix}.

THe algorithm for allreduce is consecutive execution of reduce\ to an
arbitrary root and broadcast\ from this root. Some consideration
has to be given the case of a failed root. The short version is that
the fault tolerance of reduce\ and broadcast\ yield the required
properties for \allreduce. The long version is detailed below.

The following notation is used. Analogous to \reduce,
an \allreduce\ operation is identified by an allreduce message
$m$. The content of $m$ is defined below. Each process
executes \emph{init\_allreduce($m$)} before sending any network message
related to the \allreduce\ operation, and \emph{deliver\_allreduce($m$)} to
signal completion of the operation to the caller.

The allreduce message does not store a root process, as there is no
dedicated root in allreduce. Instead it has the following members:

\begin{itemize}

\item The set of processes that participate in the operation.

\item A unique id.

\item The value. This is only relevant for the call to \emph{deliver},
  as no final value is available at the time \emph{allreduce} is called.

\end{itemize}

\subsection{Semantics}

\label{allreduce-semantics}

Let $n$ {processes} execute \allreduce. Let up to $f$ processes
fail, \preoperationally\ or \inoperationally. A set of at least
$f+1$ {processes} must be known to only fail \preoperationally,
not \inoperationally.

Then, if all processes execute \allreduce, the following semantics
hold.

\begin{enumerate}

\item If any process calls {\emph{deliver\_allreduce($m$)}}, all non-failed
  processes have called {\emph{init\_allreduce($m$)}}.

\item Each process calls {\emph{deliver\_reduce($m$)}} at most once
  for a given $m$.

\item Each non-failed process calls {\emph{deliver\_allreduce{$m$}}}
  eventually, if all non-failed processes called \emph{init\_allreduce($m$)}.

\item Each \allreduce\ includes the values of all non-failed processes.

\item The value of a failed process is either included at every
  non-failed process, or at none.

\end{enumerate}

This is proven in \cref{allreduce-proofs}.

\subsection{Description}
\label{allreduce-description}

\begin{algorithm}
  \DontPrintSemicolon
  \KwIn{The data contributed by this process}
  \KwOut{The result, including contributions of all non-failed processes}
  \SetKwData{Data}{data}
  \SetKwData{AllredData}{allreduce\_data}
  \SetKwData{RetVal}{result}
  \SetKwFunction{AllreduceFct}{allreduce}
  \SetKwProg{Fn}{function}{begin}{end}
  \Fn{\AllreduceFct{\Data}}{
    r $\leftarrow 0$\;
    ok $\leftarrow$ false\;
    \While{not ok}{
      \AllredData $\leftarrow$ reduce(\Data, root=r)\;
      \RetVal $\leftarrow$ broadcast(\AllredData, root=r)\;
      ok $\leftarrow$ broadcast finished successfully\;
      r $\leftarrow$ successor(r)\;
    }
    \Return{\RetVal}\;
  }
  \caption{\Allreduce}
  \label{alg_allreduce}
\end{algorithm}

As described above, this algorithm consists of a fault-tolerant reduce
to an arbitrary root, followed by a fault-tolerant broadcast of the
resulting value from this root. The algorithm is shown in
algorithm \ref{alg_allreduce}. Broadcast and \reduce\ need to satisfy the
semantics described in the paper on fault tolerant
broadcast~\cite{kuttler-corrected-trees} and \cref{reduce-semantics}.

Since there is no root process in \allreduce, the root is chosen
arbitrarily. The scheme requires the root to be chosen consistently
across all participating processes, and from the set of processes that
are known not to fail \inoperationally. The reason for this
restriction is the corresponding requirement of broadcast.

If the root does not fail \preoperationally, the operation
completes. Otherwise, the failure is consistently detected, and
a new root is chosen.

To guarantee progress, at least $f+1$ processes must eventually be tried as
root. Therefor a deterministic selection that selects enough processes
eventually is needed.

\subsection{Properties}
\label{allreduce-proofs}

\begin{theorem}
  \label{thm:allreduce-semantics}

  Given \reduce\ and broadcast that satisfy the semantics
  in \cref{reduce-semantics} and the paper on fault tolerant broadcast,
  respectively, \allreduce, as described
  in \cref{allreduce-description}, satisfies the semantics
  from \cref{allreduce-semantics}:

  \begin{enumerate}

  \item If any process calls {\emph{deliver\_allreduce($m$)}}, all non-failed
        processes have called {\emph{init\_allreduce($m$)}}.

  \item Each process calls {\emph{deliver\_reduce($m$)}} at most once
        for a given $m$.

  \item Each non-failed process calls {\emph{deliver\_allreduce{$m$}}}
        eventually, if all non-failed processes called
        \emph{init\_allreduce($m$)}.

  \item Each \allreduce\ includes the values of all non-failed processes.

  \item The value of a failed process is either included at every
        non-failed process, or at none.

\end{enumerate}
\end{theorem}
\begin{proof}
  Without loss of generality, assume that there is a process that does
  not fail until the end of the operation. Otherwise the semantics are
  trivially satisfied.

  \begin{enumerate}

  \item Calling \emph{deliver\_allreduce} implies having
    called \emph{deliver\_reduce} for the root process,
    and \emph{deliver\_broadcast} for all other processes. Thus,
    property 1 of both operations semantics yield this property.

  \item Like property 1, this is a direct consequence of the
    corresponding properties of \reduce\ and broadcast.

  \item First \reduce\ is executed which is delivered eventually,
    by property 5 of the semantics of \reduce.

    If the root has not failed, it sends a normal broadcast, and property 3
    and 5 of the semantics of broadcast gives this property.

    If the root has failed, it is detected, and a new root is
    tried. Because $f+1$ candidates are available at least one process
    does not fail, and a non-failed root is tried eventually.

  \item The value of the \allreduce\ is the value that is sent in the
    broadcast, which is received by the root in the previous
    \reduce. Property 3 of the semantics of \reduce\ thus provides
    this property.

  \item Whether the value of a {failed} {process}
    are included in the result of the \reduce\ to the root process is
    unspecified. But no intermediate state is possible, by property 4
    of the semantics of \reduce. Since a single root is used to
    determine the value of the \allreduce, its value provides a
    consistent result for all other processes.

  \end{enumerate}
\end{proof}

\begin{theorem}
  \label{thm:allreduce-number-messages}

  If the root did and does not fail, \allreduce\ as described
  in \cref{allreduce-description} requires as many messages to be
  sent between computation processes, as \reduce\ together
  with broadcast does.

  $f$ failures can increase this number at most to the $(f+1)$-fold.
\end{theorem}
\begin{proof}
If the root has not failed, \reduce\ and broadcast are executed in
succession. No additional messages are sent.

If the root failed, no additional messages are sent, and the operation
is retried with a new root. $f$ failures can lead to at most $f+1$
roots tried, and thus $f+1$ times the amount of messages.
\end{proof}

\printbibliography

\end{document}